\documentclass[10pt]{article}

\usepackage{fullpage}
\usepackage{hyperref}
\usepackage{url}
\usepackage{graphicx}
\usepackage{amssymb}
\usepackage{amsmath}
\usepackage{amsthm}
\usepackage{epsfig}
\usepackage{color}
\usepackage{algorithm}
\usepackage[noend]{algorithmic}

\newcommand{\eps}{\varepsilon}
\newcommand{\opt}{\mbox{\sc opt}}
\newcommand{\C}{\mathcal{C}}
\newcommand{\E}{\mathcal{E}}
\newcommand{\D}{{\cal D}}
\newcommand{\DD}{\D}
\newcommand{\pr}[1]{\mbox{\textbf{Pr}}\left[#1\right]}
\newcommand{\expect}[1]{\mbox{\textbf{E}}\left[#1\right]}

\newtheorem{theorem}{Theorem}

\newtheorem{lemma}[theorem]{Lemma}
\newtheorem{corollary}[theorem]{Corollary}
\newtheorem{claim}[theorem]{Claim}

\newcommand{\tuple}[1]{\left(#1\right)}

\newcommand{\disk}[1]{\mbox{\textbf{disk}}\tuple{#1}}
\newcommand{\pow}{\mbox{\textbf{wvd}}}
\newcommand{\cell}{\mbox{\textbf{cell}}}
\newcommand{\wt}[1]{\mbox{\textbf{wt}}\tuple{#1}}

\title{Approximation Algorithms for Dominating Set in 
Disk Graphs}
\author{Matt Gibson\thanks{Dept. of Computer Science, University of
Iowa, Iowa City, IA 52242, U.S.A. {\tt mrgibson@cs.uiowa.edu}.
Work was supported by NSF grant CCF 0915543.} 
\and Imran A. Pirwani\thanks{Dept. of Computing Science, University of
Alberta, Edmonton, AB T6G-2E8, Canada. {\tt pirwani@cs.ualberta.ca}.
Work was supported by Alberta Ingenuity.}}

\begin{document}

\maketitle

\bibliographystyle{plain}

\begin{abstract}
We consider the problem of finding a lowest cost dominating set in a
given disk graph containing $n$ disks. The problem has been extensively
studied on subclasses of disk graphs, yet the best known approximation for
disk graphs has remained $O(\log n)$ --
a bound that is asymptotically no better than the general case.
We improve the status quo in two ways: for the unweighted case, we show how
to obtain a PTAS using the framework recently proposed (independently)
by Mustafa and Ray \cite{MustafaR09} and by Chan and Har-Peled 
\cite{ChanH09}; for
the weighted case where each input disk has an associated rational
weight with the objective of finding a minimum cost dominating set, we give
a randomized algorithm that obtains a dominating set whose weight is
within a factor $2^{O(\log^* n)}$ of a minimum cost solution, with high
probability -- the technique follows the framework proposed recently by
Varadarajan \cite{Varadarajan10}. 
\end{abstract}

\section{Introduction}
For a set $\mathcal{D}$ of $n$ disks in the Euclidean plane, define
an intersection graph, $G=(V,E)$, thus: $V=\mathcal{D}$; $\{u,v\} \in
E \Leftrightarrow \disk{u} \cap \disk{v} \neq \emptyset$. $G$ is called
a {\em disk graph}; it is a {\em unit disk graph} when the disk radii are
identical.

Given a graph the {\em minimum dominating set} (MDS) problem is to find a 
smallest subset
$\mathcal{D'} \subseteq V$ such that every vertex is either in
$\mathcal{D'}$ or is adjacent to a vertex in $\mathcal{D'}$. On general
graphs, the problem is $(1-\eps) \ln n$ hard to approximate for any 
$\eps > 0$ under standard complexity theoretic assumptions
\cite{Feige98,ChlebikC04}, 
while a greedy algorithm yields an $O(\log n)$ approximation 
\cite{Vazirani01}.

Nevertheless, better approximations are possible for restricted domains.
For example, the problem admits a {\em polynomial-time approximation
scheme} (PTAS) for unit disk graphs and {\em growth-bounded graphs} 
\cite{HuntMRRRS98,NiebergHK08}. The problem is NP-hard on these domains
\cite{ClarkCJ90}.  However, for the disk graph case,
$o(\log n)$ approximations have remained elusive -- perhaps,
in part, because known techniques for unit disk graphs and 
solutions to other problems on disk graphs have either relied on packing 
properties \cite{HuntMRRRS98,NiebergHK08,ErlebachJS05,Chan03}, or
when packing property does not hold, as in the {\em minimum weighted
dominating set} on unit disk graphs, the fact that disk radii are
uniform \cite{AmbuhlEMN06,PanditPV09}. Erlebach and van Leeuwen recently
studied the dominating set problem on {\em fat objects}, e.g., disk
graphs, \cite{ErlebachL08}.  They note that existing techniques for disk 
graphs do not seem sufficient to solve MDS \cite{ErlebachL08}; they also
give an $O(1)$-approximation for fat objects of {\em bounded ply}.

In their recent break-through papers, Chan and Har-Peled
\cite{ChanH09}, and Mustafa and Ray \cite{MustafaR09} independently
showed how a simple {\em local search} algorithm on certain geometric
graphs yields a PTAS for some problems; Chan and Har-Peled
\cite{ChanH09} show local search yields a PTAS for maximum independent
set problem on {admissible} objects, while Mustafa and Ray
\cite{MustafaR09} show local search yields a PTAS for the minimum
hitting set problem given a collection of points and half-spaces in 
$\mathbb{R}^3$, and also for points and admissible regions in 
$\mathbb{R}^2$. They both use the {\em planar separator theorem}
to relate the cost of the local search solution with the optimum
solution. In the framework, at the crux lies the analysis of
a certain graph whose vertices are objects found by 
local search and ones that belong to an optimum solution, and whose edges 
(which are only between the two kinds of vertices) satisfy a
property relating the two solutions. They show that there exists
such a graph which is also planar. 
Mustafa and Ray \cite{MustafaR09} refer to the existence of
such a planar graph as the {\em locality condition}.

\paragraph{Results:}
Our first result is a PTAS for the minimum dominating set problem
for disk graphs via a local search algorithm, 
as in \cite{ChanH09,MustafaR09}.
Our analysis also uses the framework introduced
by these two papers. Our main new contribution is to show
the existence of a planar graph satisfying the locality condition. 
This graph turns out to be the dual of a weighted Voronoi diagram
in the plane. 

The minimum dominating set problem for disk graphs can be reduced 
to the problem of hitting half-spaces in $\mathbb{R}^4$ with the smallest
number of a given set of points. That is, given the set $\D$ of disks
that form the input to the MDS problem, we can easily compute a map
$\pi$ from $\D$ to a set of points in $\mathbb{R}^4$, and a map $h$ from
$\D$ to a set of half-spaces in $\mathbb{R}^4$, with the following property:
Two disks $d_1$ and $d_2$ from $\D$ intersect if and only if $\pi(d_1)$
lies in $h(d_2)$. Thus we can efficiently reduce the MDS problem
for disks to a hitting set problem for points and half-spaces in
$\mathbb{R}^4$. While there is a PTAS for the hitting set problem in 
$\mathbb{R}^3$, as shown by \cite{MustafaR09}, there is none known for
$\mathbb{R}^4$. It is not hard to see that a local search such as the
one in \cite{MustafaR09} does not yield a PTAS in $\mathbb{R}^4$.

Rather than reduce to a hitting set problem, we are able to establish
the locality condition by staying in the plane itself. In fact, the
graph for the locality condition is the dual of the weighted Voronoi diagram 
of the centers of the disks in the local search solution and the optimal
solution, where the weights are the radii of the disk. This can be
seen as generalizing the situation considered by \cite{MustafaR09} for
the hitting set problem with points and disks in the plane. In that
case, the graph for the locality condition is the Delaunay triangulation,
which is the dual of the unweighted Voronoi diagram. 

For the case when the disks are weighted, we give the first
$o(\log n)$ approximation algorithm; we give a $2^{O(\log^* n)}$
approximation algorithm\footnote{$\log^*n$ is the fewest number of
iterated ``logarithms" applied to $n$ to yield a constant.}. This result
is based on the framework recently introduced by Varadarajan for
the weighted geometric set cover problem \cite{Varadarajan10}. 
Our contribution here is to observe that the framework is applicable to our 
dominating set problem as well; the weighted Voronoi diagram is
the key to this result also. 

We assume that the inputs for both
problems satisfy non-degeneracy assumptions -- no three disk
centers on a line and no four disks tangent to a circle. This is
without loss of generality, as these conditions can be enforced by
simple perturbations. In Section \ref{sec:dom}, we present our 
PTAS for the unweighted dominating set problem, and in Section
\ref{sec:weighted} our algorithm for weighted dominating set.
 
\section{The Unweighted Case: PTAS via Local Search}
\label{sec:dom}

In this section, we give our PTAS for minimum dominating set for disk graphs.  Here, we are given a disk graph with a set $\DD$ of $n$ disks in the Euclidean plane, and we are interested in computing a minimum cardinality dominating set of the disk graph.  The algorithm is given in Section \ref{sec:algo} and the analysis of the approximation ratio is given in Section \ref{sec:approx}.

\subsection{The Algorithm}
\label{sec:algo}

\paragraph{Local Search.}
Call a subset of disks, $B \subseteq \DD$, $b$-locally optimal if one cannot obtain a
smaller dominating set by removing a subset $X \subseteq B$ of size at 
most $b$ from $B$ and replacing that with a subset of size at most $|X|-1$
from $\mathcal{D} \setminus B$.  Our algorithm will compute a $b$-locally optimal set of disks for $b = \frac{c}{\epsilon^2}$ where $c > 0$ is a large enough constant.  Our algorithm begins with an
arbitrary feasible set of disks 
and proceeds by making small local
exchanges of size $b = O(\frac{1}{\epsilon^2})$, for a given $\epsilon > 0$. 
We stop when no further local improvements are possible. 

Suppose that the solution returned
is $B$. Finally, for reasons apparent in the analysis, we check to see 
if for any disk $u \in B$ there is a disk 
$v \in \mathcal{D}$ such that $u$ is completely contained in $v \in \DD \setminus B$.  If such a disk exists, then simply
replace $u$ with the largest such disk $v$.  We return this as our final 
solution and call it $B$. Our replacement step ensures that there is no
disk in $B$ that is properly contained in some other disk in $\D$.

\paragraph{Running Time.}
We will now show that the running time can be bounded by a polynomial in $n$.  The number of swaps that the local search algorithm will make is at most $n$, because there are $n$ disks and each swap strictly decreases the number of disks in the solution.  For each swap, we need to check every subset of disks of size at most $b$ which can be done in time $O(n^b)$.  Recall that $b$ is only a function of $\epsilon$, and thus we can have the exponential dependence on $b$ in the running time.

So we will make at most $n$ swaps, each of which takes time $O(n^b)$.  Clearly, the last step (where we check to see if a disk is contained within another) can be done in time polynomial in $n$, and therefore the entire running time of the algorithm is efficient with respect to $n$.  

\subsection{Approximation Ratio}
\label{sec:approx}
We will show that our algorithm is a PTAS, thus proving the following theorem:

\begin{theorem}
 \label{thm:main}
For any $\epsilon > 0$, there exists a polynomial time algorithm for the minimum dominating set problem on disk graphs that returns a solution whose cost is at most $(1 + \epsilon)OPT$ where $OPT$ is the cost of an optimal solution.
\end{theorem}

Let $R$ be the disks in an optimal solution; we may assume no disk in $R$ is 
properly contained in any other disk in $\DD$. Thus, no disk in
$R \cup B$ is properly contained in any other disk of $R \cup B$.  
Note that by the definition of PTAS, we need to show that $|B| \leq (1 + \epsilon)\cdot |R|$.  We will refer to $R$ as the set of red disks and $B$ as the set of blue disks.  Without loss of generality, we will assume that $R \cap B = \emptyset$, i.e. there is no disk that is both red and blue.  For a disk $u \in \DD$, we say a disk $v \in R \cup B$
is a \textit{dominator of u} if $u$ and $v$ intersect.  Similarly, we also say that $v$ \textit{dominates} $u$.

We must show the existence of an appropriate planar graph which relates the disks in $R$ with the disks in $B$. Here, we state the {\em locality condition} as per Mustafa and Ray \cite{MustafaR09}:

\begin{lemma}[{\bf Locality Condition}]
 \label{lem:locality}
There exists a planar graph with vertex set $R \cup B$, such that for every $d \in \DD$, there is a disk $u$ from amongst the red dominators of $d$ and a disk $v$ amongst the blue dominators of $d$ such that $\{u,v\}$ is an edge in the graph.
\end{lemma}

Section \ref{sub:graph} is devoted to a proof of Lemma
\ref{lem:locality}.  We then describe the argument 
(from \cite{MustafaR09}) that uses the lemma to show that 
$|B| < (1+\epsilon)|R|$.

\subsection{Establishing the Locality Condition}
\label{sub:graph}
This section is devoted to the proof of Lemma \ref{lem:locality}, that is,
the construction of an appropriate planar graph which satisfies the locality
condition.

\paragraph{Weighted Voronoi Diagram.}
We will be using a generalization of Voronoi diagrams called a \textit{weighted Voronoi Diagram} (WVD).  Instead of defining cells with respect to a set of points, we will be defining cells with respect to red and blue disks.  In order to 
do this generalization for disks, we must define the distance between a point in the plane and a disk.

Let $u$ be a disk and let $x$ be a point in the plane.  We define $\pow(x,u) = d(x,c_u) - r_u$ where $c_u$ is the center of $u$, $r_u$ is the radius of $u$, and $d(x,c_u)$ is the Euclidean distance between $x$ and $c_u$.  Intuitively, for a point $x$, $\pow(x,u)$ is the Euclidean
distance from $x$ to the boundary of u; the distance to a disk is 
negative for points that are strictly inside the disk.  Alternatively, if
$x \not\in u$, then $\pow(x,u)$ is the amount we would need
to increase the radius of $u$ so that $x$ lies on the boundary of 
$u$; if $x \in u$, then $\pow(x,u)$ is the negative of the 
amount we would need to decrease the radius of $u$ so that $x$ lies on 
the boundary of $u$.  See Figure \ref{fig:pow} for an illustration.  

\begin{figure}[htpb]
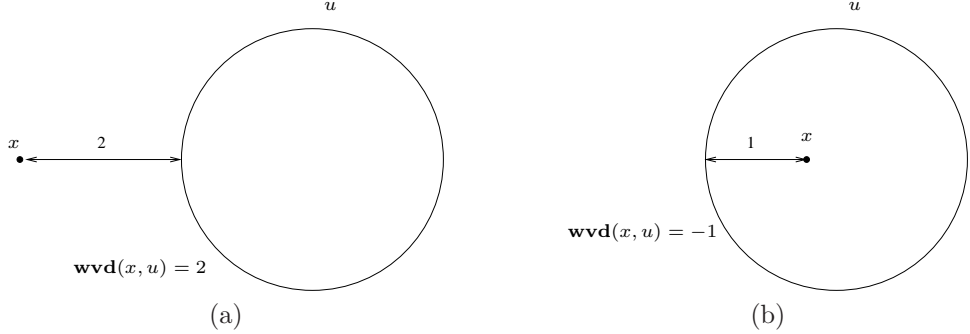

\centering
\begin{tabular}{c@{\hspace{0.1\linewidth}}c}
\input{distance3.pstex_t} &
\input{distance6.pstex_t} \\
(a) & (b)
\end{tabular}
\caption{An illustration for the distances used in our WVD.  (a) $\pow(x,u)$ when $x$ is not in $u$.  (b) $\pow(x,u)$ when $x$ is in $u$.}
\label{fig:pow}
\end{figure}

For a disk $u$ in
any collection of disks, let $\cell(u)$ be the set of points $x$ in the plane such
that $\pow(x,u) \leq \pow(x,v), u \neq v$.  The cells of all the disks in the
collection induce a decomposition of the plane, and this is the WVD. This
is just the standard weighted Voronoi diagram of the centers of the disks,
where the weight of the center of a disk is simply the radius of the disk
\cite{a91}.

Consider the WVD of the disks in $R \cup B$.
First, we will show that for every $u \in R \cup B$, $u$ has a non-empty
cell in the WVD.  That is, there is some point in the plane that is 
closer to $u$ than it is to any other red or blue disk.

\begin{lemma}
In the weighted Voronoi diagram of the union of red and
blue disks, the cell of every disk $u$ is nonempty.  Moreover, $c_u$ (the center of $u$) belongs only to $\cell(u)$.
\label{lem:nonemptycell}
\end{lemma}

\begin{proof}
We will show that $c_u$ is only in $\cell(u)$.  Suppose for the sake of contradiction that $c_u \in \cell(v)$ such that $u \neq v$.  This means that 
$\pow(c_u,v) \leq \pow(c_u,u) = d(c_u, c_u) - r_u = -r_u$.  So, $-r_u \geq 
\pow(c_u, v) = d(c_u, c_v) - r_v \Rightarrow r_v \geq d(c_u, c_v) + r_u$.  
This implies that $u$ is contained in $v$, and since the two disks
are not the same, the containment is proper. But this is a contradiction, since
no disk in $R \cup B$ contains another such disk. 
\end{proof}

\paragraph{The Graph.}

Any cell in the WVD of $R \cup B$ is star-shaped with respect to the
center of the corresponding disk.  That is, for every point $y \in \cell(u)$, 
the segment $\overline{c_u y}$ is contained within $\cell(u)$.

The graph for the locality condition is simply the dual of the WVD of
$R \cup B$.  That is, for each cell in the WVD there is a vertex, and there is 
an edge between two vertices if and only if their corresponding cells share a 
boundary in the  diagram (that is, if and only if there is a point
in the plane equidistant from the two disks). The graph is planar -- exploiting
the fact that the cells are star-shaped, the edges can easily be drawn 
so that no two edges intersect \cite{a91}.  

\begin{corollary}
The dual of the power diagram of $R \cup B$ is a planar graph.
\label{cor:dual}
\end{corollary}

Because every red and blue disk has a nonempty cell in the WVD, every such disk will also have a corresponding vertex in our planar graph.  We are now ready to show that for each $d \in \DD$, there is a disk $u$ from amongst the red dominators of $d$ and a disk $v$ amongst the blue dominators of $d$ such that $\cell(u)$ and $\cell(v)$ share a boundary in the WVD.  This would then imply that their corresponding vertices in the graph share an edge, completing the proof of Lemma \ref{lem:locality}.  For simplicity, if there is an edge connecting the vertex corresponding to $\cell(u)$ and the vertex corresponding to $\cell(v)$, then we will simply say there is an edge connecting $u$ and $v$.

\begin{lemma}
In the dual graph of the weighted Voronoi diagram for $R \cup B$, for an arbitrary
input disk $u \in \DD$, there is an edge between some red
dominator of $u$ and some blue dominator of $u$.
\label{lem:rbedge}
\end{lemma}

\begin{proof}
 Consider the WVD of $R \cup B$.  Without loss of generality, assume
 $c_u \in \cell(r)$ for some $r \in R$.  Now, $r$ must be a dominator of $u$, because $r$ is the closest disk in $R \cup B$ to $c_u$.  If $r$ does not dominate $u$, $u$ is not dominated by any disk in $R \cup B$ which contradicts the fact that both $R$ and $B$ are dominating sets.  

Let $b$ denote a closest blue disk to $c_u$, that is $\pow(c_u,b) \leq \pow(c_u, b')$ for all other blue disks $b'$.  Note that $b$ must dominate $u$, because if it did not, then no blue disks would dominate $u$.  This would contradict the fact that $B$ is a dominating set.  Also, note that for any disk $d \in \DD$ such that $\pow(c_u,d) \leq \pow(c_u,b)$, $d$ must intersect with $u$.  

If $\pow(c_u,b) = \pow(c_u,r)$, we are done, since then there is an
edge in the dual graph incident on $r$ and $b$. So, let us assume
that $\pow(c_u,b) > \pow(c_u,r)$.

We will walk from $c_u$ to $c_b$ along the straight
line segment $\overline{c_uc_b}$. The proof strategy is that during this
walk, we will be crossing red cells and at some point 
before reaching $c_b$ we will enter a
blue cell, in particular, $\cell(b)$.  We must have
entered this cell from a red cell $\cell(r')$ which shares a boundary 
with $\cell(b)$, and thus $\{r',b\}$ is an 
edge in our planar graph.  Moreover, we will argue that $r'$
necessarily dominates $u$, completing the proof.

As seen in the proof of Lemma~\ref{lem:nonemptycell}, $c_b
\in \cell(b)$, and thus we will enter $\cell(b)$ at some point in time 
along our walk from $c_u$ to $c_b$. 
Let $x$ be the point at which we first enter $\cell(b)$.  Then $x$ is on the
boundary of $\cell(b)$ and $\cell(r')$ for some $r' \in R \cup B$. If
$r' = r$, we are done. Otherwise, we have

$$\pow(c_u, r') < d(c_u, x) + \pow(x,r') = d(c_u,x) + \pow(x,b) = \pow(c_u,b).$$

\noindent (Here the strictness of the first inequality comes from our
non-degeneracy assumption which implies that $c_{r'}$ cannot lie on the
line through $c_u$ and $c_b$.) Now, it must be the case that $r' \in R$ because 
$\pow(c_u,r') < \pow(c_u,b)$ and $b$ is the closest blue disk to $c_u$.  This 
also implies that $r'$ must dominate $u$.  See Figure~\ref{figure:rbedge} for 
an illustration.


\begin{figure}[h]
\centering
\includegraphics[height = 2in]{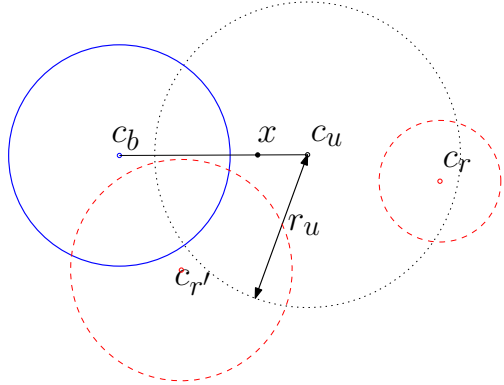} 
\caption{Proof of Lemma~\ref{lem:rbedge}. The dotted disk is $u$ with
center $c_u$ and radius $r_u$. The two red disks $r$ and $r'$ are shown
as dashed disks with centers $c_r$ and $c_{r'}$, respectively. The only
blue disk $b$ is shown as a solid disk with center $c_b$.}
\label{figure:rbedge}
\end{figure}

Therefore $\cell(b)$ and $\cell(r')$ share a boundary implying that the edge
$\{b,r'\}$ is in our graph.  Moreover, $b$ is blue, $r'$ is red, and both
dominate $u$, which completes the proof. 
\end{proof}

Together, Corollary~\ref{cor:dual} and Lemma~\ref{lem:rbedge} prove
Lemma~\ref{lem:locality}.

\subsection{Proof of Theorem \ref{thm:main}}
To show $|B| \leq (1 + \epsilon)\cdot |R|$, we make use of the planar graph separator theorem of Frederickson \cite{Frederickson87}.  This argument is similar to the work in \cite{ChanH09,MustafaR09} and is only given here for completeness.  Given a graph $G = (V,E)$, we denote $N(V')$ for subset of the vertices $V'$ to be the set of all vertices in $V$ that share an edge with a vertex in $V'$.

\begin{theorem}[Frederickson \cite{Frederickson87}]
There are constants $c_1, c_2, c_3 > 0$, such that for any planar graph
$G=(V,E)$ with $n$ vertices and a parameter $r \geq 1$, there is a set
$X \subseteq V$ of size at most $c_1 n/\sqrt{r}$, and a partition
of $V \setminus X$ into $n/r$ sets $V_1, V_2, \ldots, V_{n/r}$,
satisfying: (i) $|V_i| \leq c_2 r$, (ii) $N(V_i) \cap V_j = \emptyset$, for $i \neq j$, and
(iii) $|N(V_i) \cap X| \leq c_3 \sqrt{r}$.
\label{thm:septhm}
\end{theorem}

We will now show that $|B| \leq (1+\epsilon) |R|$;
this is similar to \cite{ChanH09,MustafaR09,GibsonKKV09}.  Let $r \equiv b/(c_2+c_3)$ (where $b$ is the parameter from the local search algorithm).
From Theorem~\ref{thm:septhm}(i),(iii), we get $|V_i \cup N(V_i)| \leq c_2r + c_3\sqrt{r} \leq b$.
Let $R_i = R \cap V_i$ and $B_i = B \cap V_i$. Due to the optimality of local search, we must have
$|B_i| \leq |R_i| + |N(V_i)|$, otherwise local search can replace $B_i$ with $R_i \cup N(V_i)$ to obtain
a smaller dominating set, contradicting the local optimality of local search. This is why we require that for each $d \in \DD$, there is a red dominator of $d$ and a blue dominator of $d$ with an edge in the graph.  If there were no such edge, then making this swap could possibly leave some disks without a dominator.  So now we have,
\begin{eqnarray*}
  |B| & \leq & |X| + \displaystyle\sum_i |B_i| \leq |X| + 
		\displaystyle\sum_i |R_i| +
		\displaystyle\sum_i |N(V_i)| \leq |R| + c \frac{|R| + |B|}{\sqrt{r}} \\
	& \leq & |R| + c' \frac{|R| + |B|}{\sqrt{b}},
\end{eqnarray*}
where $c$ and $c'$ are positive constants. With $b$ a large enough constant times $1/\epsilon^2$, it
follows that $|B| \leq (1+\epsilon) |R|$.

\section{The Weighted Dominating Set Case}
\label{sec:weighted}
In this section, we study a classical generalization of the dominating
set problem. Each disk $u$ now has an associated rational weight, $w_u$. 
The goal is to find a dominating set $D$ having the lowest cost, that
is, $\wt{D} = \sum_{u \in D} w_u$ be as small as possible.
We will prove the following theorem:

\begin{theorem}
Given a disk graph, $G=(V,E)$ of $n$ weighted disks $D$ in the 
plane, there is a randomized algorithm that produces a dominating set $V'
\subseteq V$, and $\wt{V'} \leq 2^{O(\log^*n)} \cdotp \opt$, w.h.p.,
where $\opt$ denotes the cost of an optimal solution.
\label{thm:mainthm}
\end{theorem}

The high-level structure of the algorithm is as follows: 
we first solve a natural linear 
programming relaxation, followed by a randomized rounding step; this
step allows us to ignore the weights of the disks in the sampling
(pruning) stage. 
In the rounding step, we make several copies of the disks to ensure
that two properties hold.  First, every disk in $D$ is covered by 
at least $n$ of the copies.  Second, the weight of the copies is 
$O(n \cdot \lambda^*)$, where $\lambda^*$ is the objective function 
value of an optimal LP solution.  Following this step, we
recursively apply a randomized pruning step where we remove some of the copies
according to the algorithm given in the proof of Theorem~\ref{thm:sparsify} 
while ensuring that the remaining copies are a dominating set of $D$.  The main goal
of the pruning step is to remove some of the copies while approximately preserving
the ratio of the cost of the remaining copies to the ``depth'' of the disks
in $D$ with respect to the remaining copies.
We recursively apply the pruning step until the disks in $D$ are covered by
only a constant number of the remaining copies; the depth of our recursion is 
$\Theta(\log^* n)$.  We can then show that the expected weight of our final 
dominating set is at worst $2^{O(\log^* n)}\cdot \lambda^*$.

First, we define some terms that are used in the remaining part of the
section. Given a disk $v$ and a set of disks $S$, 
we say that $v$ is $L$-covered by $S$ if there are exactly 
$L$ disks in $S$ each of which intersects
$v$. In other words, neighborhood of $v$ in $S$ has size $L$.
We will make use of the following lemma, which is our main contribution to 
the weighted case:

\begin{lemma}
Let $S$ be a set of $m$ disks, and $1\leq L\leq m$ an integer. 
Let $Q$ be another (possibly infinite) set of disks.
There are $O(m \cdotp L^2)$ disks of $Q$ that intersect distinct
subsets of $S$ each of size at most $L$.
\label{lem:intersectionlemma}
\end{lemma}
\begin{proof}
We first define a few concepts that we use in the proof.
We focus on subsets $S'\subseteq S$ of size at most $L$ and disks of $Q$
whose neighborhood is precisely one of these subsets; let us denote
this subset of $Q$ by $Q'$.
For a set $S' \subseteq S$ of size at most $L$, and a pair of disks 
$u,v\in Q'$, 
we say that $u$ and $v$ are related if they both intersect every disk in $S'$ 
and
no other disk of $S \setminus S'$, i.e. $u \cap S = v \cap S = S'$. 
So we have an equivalence relation on $Q'$ where each 
equivalence class corresponds
to a set $S'\subseteq S$. 
We wish to bound the number of these equivalence classes.
Let these subsets of $S$ be
$\{S_1, S_2, \ldots, S_t\}$, and correspondingly, $t$ equivalence
classes $\{Q_1, Q_2, \ldots, Q_t\}$, where each disk in $Q_i$ intersects
every disk in $S_i$, and no other disk of $S \setminus S_i$. 
Consider any set $Q_i$ and an arbitrary disk $v\in Q_i$. By scaling and/or
translating $v$ we can obtain a disk $v'$ with the following property: $v'$ 
has the same neighborhood as all the disks in $Q_i$ and is sharing a single 
point 
with three, two, or one disk in $S_i$ and is intersecting all the other disks 
in more than one point; for the cases when $v'$ is touching a single 
disk in $S_i$, or two disks in $S_i$, we continue to translate and scale
$v'$ so that it touches two disks outside of $S_i$, or one disk outside
of $S_i$, respectively. Without loss of generality, we assume that $S$
has four special disks whose borders form the North, South, East, and
West boundary, respectively, of the region that contains the input
disks. We call these special disks $N,S,E,W$, respectively.  
Such a transformed disk, $v'$, that touches exactly three disks
is referred to as $v_i$.
We say that a disk $d$ is {\em canonical} with respect to a set of disks
$D'$ if there are three distinct disks in $D'$ such that $d$ intersects
the three disks at only one point each.  Note that each $v_i$ is a
canonical disk with respect to the set $S$.  We say that a canonical
disk $v$ is \textit{$\kappa$-canonical} with respect to a set of disks $D'$ if
at most $\kappa$ disks from $D'$ intersect the interior of $v$.  Therefore,
each of the canonical disks $v_i$ that we defined are $L$-canonical disks.
It is easy to see that $t$ is within a constant factor of the number 
of $L$-canonical disks with respect to $S$.  
For each $v_i$, the set of disks that shares exactly one point with
it is called the {\em defining set} of $v_i$ and every disk of $S_i$ that
shares more than one point with $v_i$ is said to be in the {\em conflict
set} of $v_i$. Note that the defining set of $v_i$ has at least one disk
from $S_i$, but at most two remaining disks can be from outside $S_i$.
We will upper bound the number of $L$-canonical disks with respect to $S$
(and hence upper bound $t$) by choosing a random sample $S' \subseteq S$ and 
calculating the expected number of $0$-canonical 
disks with respect to $S'$. This technique dates back 
to that of Clarkson \cite{Clarkson88}.


We pick a random sample $S' \subseteq S$, such that
$\pr{d \in S'} = \frac{1}{L}$, while $N,S,E,W \in S'$ with probability 1.  
Let $X_d$ be an indicator variable
denoting the event that $d$ is $0$-canonical for $S'$. So,
$$\pr{X_d = 1} \geq \left(\frac{1}{L}\right)^3 \cdotp
										\left(1-\frac{1}{L}\right)^j \geq
										\left(\frac{1}{L}\right)^3 \cdotp
										\left(1-\frac{1}{L}\right)^L \geq \frac{1}{eL^3}$$
where, $d$ is $j$-canonical with respect to $S$ and $j \leq L$.

We will show that the maximum number of $0$-canonical disks for $S'$ is
$O(|S'|)$.
\begin{claim}
For a set $S'$ of disks of size $k$, the maximum number of $0$-canonical 
disks induced is $O(k)$.
\label{linearbound}
\end{claim}
\begin{proof}
We will bound the number of $0$-canonical disks by the number of Voronoi
vertices of a weighted Voronoi diagram with $k$ sites in which the
sites are represented by the $k$ centers of disks in $S'$, and the
weight of each site is the radius of the corresponding disk.
Every Voronoi vertex is equidistant from the disks of the regions sharing
that vertex. So each Voronoi vertex in the Voronoi diagram corresponds 
to the center of a disk that touches
the boundary of exactly three disks of $S'$ (disks corresponding to the three
regions defining that vertex) and does not intersect any other disk of $S'$.
Since the
number of Voronoi vertices of a Voronoi
diagram having $k$ sites is bounded linearly in $k$, the number of of
canonical disks that touch three disks of $S'$ are thus bounded
linearly in $k$ as well.  This leads to the final bound of $O(k)$ 
on the maximum number of canonical disks that $S'$ admits. 
\end{proof}
According to the claim, the maximum number of $0$-canonical disks for
$S'$ is $O(|S'|)$. So,
$$\sum_d X_d \leq O(|S'|) \Rightarrow
	\expect{\sum_d X_d} \leq \expect{c_0 |S'|} = c_0 \frac{m}{L}, 
	c_0 \in O(1)$$
On the other hand,
$$\expect{\sum_d X_d} = \sum_d \expect{X_d} = \sum_d \pr{X_d = 1}
										\geq t \frac{1}{eL^3}$$
Thus, $t \leq c' m L^2$, where $c' \geq c_0 e \in O(1)$.

\end{proof}
We prove the following variant of a theorem of Varadarajan in 
\cite{Varadarajan10}.
\begin{theorem}
Given a disk graph $G=(V,E)$ and set of $n$ weighted disks $D \subseteq V$ in the plane s.t. 
$D$ dominates $V$, there is a randomized algorithm that produces a subset $D' \subseteq D$, 
such that for any disk $v \in V$, if $v$ is 
$L$-covered in $D$, then $v$ is at least $\log L$-covered in $D'$ and
$\pr{d \in D'} \leq \frac{c\cdotp\log L}{L}$.
\label{thm:sparsify}
\end{theorem}
\begin{proof}
We only describe a randomized process that selects a subset, $D'$ of disks
such that any disk $v \in V$ that is covered by $D$ in the  
range $[L,2L]$, $v$ is at least $\log L$-covered in $D'$. 
Let $N_m = D$, and let $\C_m$ denote the set of equivalence classes of 
disks in $V$ such that each class intersects at most 
$2L$ disks of $D$. Note that since the disks in one equivalence class of $V$ have
the same neighborhood in $D$, if we obtain a set $D'$ that at least $\log L$-covers one disk
in that class, then all the disks in that class are also at least $\log L$-covered.
Therefore, we can assume we have one representative disk from each class and our goal
is to at least $\log L$-cover these disks. We use this fact crucially in our analysis.
By Lemma~\ref{lem:intersectionlemma},
$|\C_m| \leq c' \cdotp n_m L^2$, $n_m = |N_m|$. So, there is a disk $d_m$ that
covers at most $2c' L^2$ classes of $\C_m$.  Find such a disk $d_m \in N_m$,
and recursively compute a sequence for $N_{m-1} = N_m\setminus \{d_m\}$, and
append the sequence to $d_m$. That is, in the arrangement of $N_{m-1}$ we 
consider
the classes $\C_{m-1}$ whose coverage in $N_{m-1}$ is at most $2L$.
The recursion stops when there are fewer
than $L$ disks remaining, at which point, we compute an arbitrary
sequence of the remaining set of disks.

Let $\sigma$ be the {\em reverse} of this sequence, that is,
$\sigma = (d_1,d_2,\ldots, d_m)$. When considering disk $d_j$, we make
an instant decision about including it in our cover or not. Call a disk
$d_j \in N_j$ {\em forced} if for some disk $v \in \C_j$, 
not including $d_j$ {\em will not} $\log L$-cover $v$, 
whose coverage in $N_m$ is in
$[L,2L]$. Otherwise, if $d_j$ is not forced, we add it to $D'$ with
probability $\frac{c \cdotp \log L}{L}$.  We will upper bound the probability 
of $d_j$ being forced -- we will show that it is at most $O(1/L)$. 

Observe that if a disk $d_j$ is forced because of $v$, then all the
disks $d_{j'}$ (with $j'\geq j$) that cover $v$ are also forced, and the
number of such disks is at most $\log L -1$ (otherwise $d_j$ won't be forced).
So it is sufficient to upper bound the probability of a disk $d_i$ being the 
first disk forced because of $v$. Let us denote this event by $\E_i(v)$.
Since from among the disks that cover $v$ at most the last $\log L$ 
disks
can be forced, the probability of one of these $\log L$ disks being forced is
at most $\log L$ times the probability that one of the disks before it is the
``first'' forced disk because of $v$.
We use $\E_i$ to denote the event that $d_i$ is the first disk forced because of
some disk that it covers.

Consider a disk $v$ whose coverage is in the range $[L,2L]$ and $d_i$ covers it.
To bound the probability of $\E_i(v)$
consider the subsequence of $\sigma$, called $\sigma'$, that covers
$v$ (so coverage of $v$ is still in $[L,2L]$). We will only focus
on this subsequence since it is the only one which is relevant to
$\E_i(v)$.
Observe that the length of $\sigma'$ is in the range $[L,2L]$.
The probability that $v$ forces any disk
will then be the sum of the probabilities of any of the last $\log L$
disks in $\sigma'$ being a ``first'' forced disk (by $v$).
Let the prefix of $\sigma'$ leading
up to $d_i$ be called $\sigma_i'$, that is, 
$\sigma_i' = (d_{a_1},d_{a_2},\ldots, d_{a_l = i})$.  
If $d_i$ is going to be the first forced disk in $\sigma'$, then
$i \geq L-\log L$. Suppose $v$ is the representative disk of one of the $O(L^2)$ classes
of $\C_i$ that are covered by $d_i$.
Since $\sigma_{i-1}'$ is a subset of disks that cover $v$,
$\sigma_{i-1}$ is also the set of disks that cover $v$. We will
bound the probability that an insufficient number of disks were
picked from $\sigma_{i-1}'$, leading to the forced inclusion of
$d_i$ as the first forced disk for $v$.
Let $x$ be the number of disks picked from
$\sigma_{i-1}'$ and fix a particular choice of $x$ disks from
$\sigma_{i-1}'$, $0 \leq x < \log L$. 
\begin{eqnarray*}	
	\pr{
		\substack{
			\text{$x$ particular} \\
			\text{disks from $\sigma'_{i-1}$}\\
			\text{are picked}
		} 
		\wedge
		\substack{
			i-x-1\\
			\text{disks are}\\
			\text{dropped}
		}
	}
	=
	\pr{
		\substack{
			\text{$x$ particular} \\
			\text{disks from $\sigma'_{i-1}$}\\
			\text{are picked}
		} 
	} 
	\cdotp
	\pr{
		\substack{
			i-x-1 \\
			\text{disks are}\\
			\text{dropped}
		}
	}
	=\\
	\left(\frac{c \cdotp \log L}{L}\right)^x \cdotp
			\left(1 - \frac{c \cdotp \log L}{L}\right)^{i-x-1}
\end{eqnarray*}

Since $i \geq L-\log L$ and there are $i-1 \choose x$ ways to choose $x$ disks
from $\sigma_{i-1}'$:
\[
\begin{array}{rcl}
	\pr{\E_i(v)}  &\leq& {i-1 \choose x} 
			\left(\frac{c \cdotp \log L}{L}\right)^x \cdotp
			\left(1 - \frac{c \cdotp \log L}{L}\right)^{L-2\log L-1}\\
		& \leq &
			{L-1 \choose x} \cdotp
			\left(\frac{c \cdotp \log L}{L}\right)^x \cdotp
			\left(1 - \frac{c \cdotp \log L}{L}\right)^{\frac{L}{2}} \\
		& \leq &
			\left(\frac{e \cdotp (L-1)}{x}\right)^x \cdotp
			\left(\frac{c \cdotp \log L}{L}\right)^x \cdotp
			\left(1 - \frac{c \cdotp \log L}{L}\right)^{\frac{L}{2}}\\
		& = &
			\left(\frac{c \cdotp e \cdotp (L-1) \cdotp \log L}{L \cdotp x}\right)^x \cdotp 
			\left(1 - \frac{c \cdotp \log L}{L}\right)^{\frac{L}{2}}
\end{array}
\]
\begin{claim}
\[
	\max_{0<x<\log L}
	\left\{\left(\frac{c \cdotp e \cdotp (L-1) \cdotp \log L}{L \cdotp x}\right)^x \right\} 
	< L^4
\]
\label{firstderivtest}
\end{claim}
\begin{proof}
We analyze the continuous function, 
$\left(\frac{c \cdotp e \cdotp (L-1) \cdotp \log L}{L \cdotp x}\right)^x$, 
and apply the
{\em first derivative test} from elementary calculus with respect to $x$, 
$x \in \mathbb{R}$.
Note that the only value for which the function maximizes 
$x = \frac{(L-1)\log L^c}{L} > \log L$. Restricting the domain, it
follows that it maximizes for the boundary point, $\log L$, yielding
a value of less than $L^4$, when $c \in O(1)$. 
\end{proof}
By applying the claim, in the case when $x > 0$, 
\[
\begin{array}{rcl}
	\pr{\E_i(v)}  
		& \leq & 
		L^4 \cdotp \frac{1}{L^{c/2}} \leq 
		\frac{1}{L^4}, \text{ for } c \geq 16
\end{array}
\]

Note that any disk $d_{i'}$ that occurs before $d_i$ in $\sigma$
if $d_{i'}$ is forced for a disk $v'$ that is not covered by $d_i$, which 
forces a disk $d_k$ which occurs after $d_i$ in $\sigma$ and that $d_k$
also covers $v$, then that event has no bearing on the event of
$d_i$ being a first forced disk for $v$
(see Figure~\ref{fig:fig1}).
\begin{figure}[h]
\centering
\includegraphics[width = 4in]{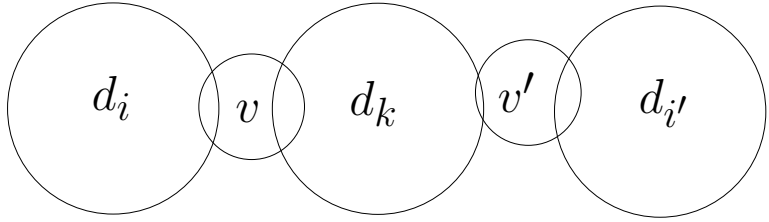}
\caption{$d_i$ occurs after $d_{i'}$ and before $d_k$ in $\sigma$.
$d_{i'}$ was forced for $v'$, so $d_k$ is also forced for $v'$. However,
$d_k$ also covers $v$. Suppose $d_i$ becomes forced for $v$. The
fact that $d_k$ is already in $D'$ does not increase the probability
of $d_i$'s forced inclusion for $v$.}
\label{fig:fig1}
\end{figure}

So, to upper bound the probability that some $d_j$ is a forced
disk for a fixed disk $v$, we sum over all valid indices $i<j$ with $d_i$
being the first forced disk because of $v$, and 
obviously there are at most $\log L$ of them,
\[
	\pr{
		\substack{
			\text{some $d_j$ is} \\
			\text{forced by}\\
			\text{disk } v
		}
	} \leq \displaystyle\sum_{i}
					\frac{1}{L^4} \leq \frac{1}{L^3}.	
\]
Since there are at most $2c'L^2$ classes of $\C_j$ 
having coverage in the range $[L,2L]$ that are covered by $d_j$,
$d_j$ can be a forced addition for any one of the at most $2c' L^2$ representative disks. 
So,
\[
	\pr{ \substack{
                     d_j \text{ is forced}\\
                     \text{for some } \\
                     \text{disk } v \in \C_j
                }
            } 
           \leq \frac{2c'}{L}.
\]
The probabilistic algorithm finds a dominating set $D' \subseteq D$  
where the probability of a given disk being in $D'$ is at most
$\frac{c \cdotp \log L}{L}$ and each disk $v \in V$ that is covered 
in the range $[L,2L]$ by $D$, is at
least $\log L$-covered in $D'$. We repeat the process
for points that are between $2L$ and $4L$ deep, and so on. Note that the 
probability of a disk being in $D'$ is still the same. 
\end{proof}

\subsection{Proof of Theorem~\ref{thm:mainthm}}
Let the input instance be a disk graph based on a set of disks $D$. 
For any disk $d \in D$, let
$N[d]$ denote the set of neighbors of $d$ in the graph, inclusive.
Consider the following natural LP relaxation for the weighted 
dominating set problem:
\[
\begin{array}{rl}
	\text{(LP)} & \min \displaystyle\sum_{d \in D} w_d x_d \\
\end{array}
\]
subject to,
\[
\begin{array}{rclr}
	\displaystyle\sum_{d':d' \in N[d]}x_{d'} & \geq & 1, & \forall d \in D \\
	x_d & \geq & 0, & \forall d \in D
\end{array}
\]

After solving the LP relaxation, we create a set $D_0$ of disks as
follows. For each disk $d$ such that $x_d \geq 1/2n$, we add 
$\lfloor \frac{x_d}{1/2n} \rfloor$ copies of $d$ to $D_0$. Each copy
of $d$ inherits its original cost. For each disk $d$ with 
$x_d < 1/2n$, we don't add any copy to $D_0$. It is easily verified
that $\wt{D_0} \leq 2n \cdot \lambda^*$, where $\lambda^*$ is
the objective function value of the optimal LP solution. Furthermore,
we have that each disk $d \in D$ is $n$-covered by $D_0$. 

In the next phase, our algorithm will recursively apply 
Theorem~\ref{thm:sparsify}
to obtain a successively sparse dominating set. For the $i$th application
of the theorem, we set $L_i = \log L_{i-1}$, for $i = 2, 3, \ldots, t$
to obtain a set $D_i \subseteq D_{i-1}$. For the first
application, we set $L_1 = n$.  

Fix a disk $d$ and let $t$ be the final application of 
Theorem~\ref{thm:sparsify}, to yield $D_t$. Let us calculate the 
expected weight of $D_t$. Let $\E_i^d$ denote the event that
$d \in D_i$. So, 
\begin{eqnarray}
	\expect{\wt{D_t}} & = &
		\displaystyle\sum_{d \in D_0} \pr{\E_t^d} \cdotp w_d
\label{low-cost}
\end{eqnarray}
For a disk $d$, note that,
\begin{eqnarray*}
\pr{\E_t^d} 
	& = &
	\pr{\E_t^d \Big\vert
		\bigwedge_{i=1}^{t-1} \E_i^d} \cdotp  
		\pr{\E_{t-1}^d \Big\vert \bigwedge_{i=1}^{t-2} \E_i^d} \ldots
		\pr{\E_3^d \Big\vert \E_2^d \wedge \E_1^d} \cdotp 
		\pr{\E_2^d \Big\vert \E_1^d} \cdotp \pr{\E_1^d} \\
	& \leq  &
	\frac{c \cdotp \log L_{t}}{L_{t}} \cdotp
	\frac{c \cdotp \log L_{t-1}}{L_{t-1}} \cdotp
	\frac{c \cdotp \log L_{t-2}}{L_{t-2}} \cdotp
	\frac{c \cdotp \log L_{t-3}}{L_{t-3}} \cdotp \ldots
	\frac{c \cdotp \log L_i}{L_i} \cdotp \ldots
	\frac{c \cdotp \log L_1}{L_1} \\
	& =  &
	\frac{c \cdotp \log L_{t}}{\log L_{t-1}} \cdotp
	\frac{c \cdotp \log L_{t-1}}{\log L_{t-2}} \cdotp
	\frac{c \cdotp \log L_{t-2}}{\log L_{t-3}} \cdotp
	\frac{c \cdotp \log L_{t-3}}{\log L_{t-4}} \cdotp \ldots
	\frac{c \cdotp \log L_i}{\log L_{i-1}} \cdotp \ldots
	\frac{c \cdotp \log L_2}{\log L_1} \cdotp
	\frac{c \cdotp \log L_1}{L_1} \\
	& = & \frac{c^{t} \cdotp \log L_{t}}{L_1} = \frac{c^{t} \cdotp 
					\log^{(t)} L_1}{L_1}
\end{eqnarray*}
Our recursion stops at a depth $t$ when $\log L_t \leq 1$. Eventually, we get
a dominating set $D_t$, that is, each disk $v \in V$
is at least $1$-covered in $D_t$. 
So, truncating the recursion at a depth of $t = \log^* n$, we get that,
\[
	\pr{\E_t^d} \leq \frac{c^t}{L_1} 
\]
Since $L_1 = n$ and $\sum_{d \in D_0} w_d \leq 2n\lambda^*$, inequality-(\ref{low-cost}) becomes,
\[
	\expect{\wt{D_t}} \leq \frac{c^t}{L_1} \displaystyle\sum_{d \in D_0} w_d \leq c^t 2 \lambda^*
\]
Using Markov's inequality,
\begin{eqnarray*}
	\pr{\wt{D_t} > 2 \cdotp \expect{\wt{D_t}}} & \leq &
			\frac{1}{2}	
\end{eqnarray*}
Repeating $O(\log n)$ independent trials and taking the lowest cost
amongst them yields,
\[
	\pr{\wt{D_t} \leq 2^{O(\log^* n)} \cdotp \lambda^*} > 1 - \frac{1}{n} 
\]
This proves Theorem~\ref{thm:mainthm}.

\paragraph{Running Time}
The algorithm requires finding a disk $d_j$ in $D$ such that
$d_j$ intersects $O(m \cdotp L^2)$ equivalence classes
from amongst the subset of disks in $V$ whose neighborhood in $D$
has size at most $L$. We construct these classes, $\{V_1, V_2, \ldots,
V_t\}$, by simply examining a disk $v$ and placing it in $V_i$ such that
any disk $u \in V_i$ has the exact same set of neighbors in $D$, and
the size of each neighborhood is bounded by $L$.
Since $t$ is bounded by $O(|D| \cdotp L^2)$, a naive approach of
comparing the neighborhood of each disk with the neighborhood of each
of the $V_i$'s, takes $O(|V| \cdotp |D| \cdotp L^4)$ time. So naively, the
equivalence classes $\{V_1, V_2, \ldots, V_t\}$ can be constructed in  
$O(|V|^2 \cdotp |D| \cdotp L^4)$ time. Constructing a set of 
representative disks $U = \{u_i\}_{i=1}^{t}$ from
the equivalence classes takes $O(|D| \cdotp L^2)$ time. Finding
a disk $d_j \in D$ such that $d_j$ intersects $O(L^2)$ disks of $U$
takes $O(|D|^2 \cdotp L^4)$ time. So, constructing the sequence,
$\sigma$, takes $O(|D|^3 \cdotp L^4)$ time. The quasi-uniform sampling
stage takes $O(|D|^2 \cdotp L^3)$ time for each disk
$d_j$ because we need to determine if $d_j$ is forced for any of the
disks in $U$, any one of which may contain some subsequence of disks
that occur after $d_j$ in $\sigma$. Since there are a maximum of $|D|$
disks that occur after $d_j$, the running time of the sampling stage,
given $\sigma$, is $O(|D|^3 \cdotp L^3)$. So, the overall running time
is $O(|V|^2 \cdotp |D| \cdotp L^4 + |D|^3 \cdotp L^3)$ for each 
recursive application of the
procedure in the proof of Theorem~\ref{thm:sparsify}.
Since each of $|V|, L, |D|$ is bounded polynomially in $n$, the
number of input disks, and given that the depth of the recursion is
bounded by $O(\log^* n)$ for each independent random experiment, and
$O(\log n)$ bounding the maximum number of random experiments, the over
all running time is bounded polynomially in $n$, using a naive
approach that does not seek to optimize the running time.

\section{Concluding Remarks and Open Questions}
Given the negative result of Marx \cite{Marx07a} which shows that even
for the simple case of unweighted unit disk graph, an EPTAS for the
problem would contradict the {\em exponential time hypothesis} 
\cite{ImpagliazzoP01}\footnote{Marx \cite{Marx07a} 
actually shows something stronger.}, it is unlikely that the
dependence of $1/\eps$ as an exponent of $n$ on the running time for the 
PTAS can be improved to, say, $f(1/\eps) \cdotp n^{O(1)}$. However, the 
running time of the local search PTAS is $n^{O(1/\eps^2)}$. Can this 
be improved to $n^{O(1/\eps)}$? In our work, we have made no attempt to
improve the running time.

For the weighted case, we are only able to show a constant integrality
gap for the lower bound despite numerous attempts. Thus, we believe that 
the right upper bound for the approximation factor is $O(1)$.

\paragraph*{Acknowledgments:} We thank Sariel Har-Peled and Kasturi Varadarajan for suggesting the use of
weighted Voronoi diagrams for the unweighted case, and we thank Kasturi Varadarajan for
pointing out the connection between weighted set cover and
weighted dominating set. We also thank Mohammad Salavatipour for his support and many valuable discussions.

\bibliography{diskgraph}

\end{document}